\documentclass[12pt,reqno]{amsart}

\usepackage[utf8]{inputenc}

\usepackage[left = 0.8in, right = 0.8in, top = 0.8in, bottom = 0.8in]{geometry}
\usepackage{amsmath,amsthm,amssymb}
\usepackage{multicol}
\usepackage{graphicx}
\usepackage{color}
\usepackage[sorting=none, url = false, doi = false, eprint = false, isbn = false]{biblatex} 
\usepackage{cancel}
\usepackage{todonotes}
\usepackage{cleveref}
\usepackage[normalem]{ulem}

\newtheorem{theorem}{Theorem}

\newtheorem{remark}{Remark}

\newcommand{\R}{\mathbb{R}}

\newcommand{\Z}{\mathbb{Z}}
\newcommand{\C}{\mathbb{C}}

\DeclareMathOperator\im{Im}
\DeclareMathOperator\sgn{sgn}
\DeclareMathOperator{\sech}{sech}
\DeclareMathOperator{\csch}{csch}

\newcommand{\stkout}[1]{\ifmmode\text{\sout{\ensuremath{#1}}}\else\sout{#1}\fi}

\title{Bound States in Exactly Solvable Polaritonic Models}

\author[]{Sam Johar}
\address{Department of Mathematics, Amherst College, Amherst, MA 01002, USA}
\email{sjohar27@amherst.edu}

\author[]{Ryogo Katahira}
\address{Department of Mathematics, Amherst College, Amherst, MA 01002, USA}
\email{rkatahira27@amherst.edu}

\author[]{Joseph Kraisler}
\address{Department of Mathematics, Williams College, Williamstown, MA 01267, USA}
\email{jk34@williams.edu}

\date{\today}

\begin{document}

\begin{abstract}
We study a one-dimensional scalar model for an effective-mass quasiparticle field coupled to a continuum of two-level atoms with at most one excitation present. We prove a general theorem on the existence of bound states under sign-definite spatially localized perturbations of constant atomic density and analyze the bound states for several exactly solvable examples. 
\end{abstract}

\maketitle

\section{Introduction and Model}\label{sec:Intro}

In this paper we consider a one-dimensional model of a scalar quantized field with dispersion relation $\omega(k)=k^2$ coupled to a continuum of identical two-level atoms. The atoms are distributed according to a number density $\rho(x)$ and have two distinct energy levels: a ground state energy normalized to $0$ and an excited state energy $E=\hbar\Omega>0$. We refer to $\Omega$ as the resonant frequency of the atoms. At the level of quantized fields, the Hamiltonian $\mathcal{H}$ is a sum of three terms, one for the energy of the field, a second for the atomic energy, and a third corresponding to the interaction energy:
\begin{align}
    \mathcal{H} = \hbar\int_{\R} \phi^{\dagger}(x)\Delta\phi (x)dx + \hbar\Omega\int_{\R}\sigma^{\dagger}(x)\sigma(x)\rho(x)dx + \hbar g\int_{\R}\left(\sigma(x)\phi^{\dagger}(x)+\sigma^{\dagger}(x)\phi(x)\right)\rho(x)dx.
\end{align}
where $\phi^{\dagger}(x), \phi(x)$ are bosonic creation and annihilation operators for the field respectively, while $\sigma^{\dagger}(x), \sigma(x)$ are atomic raising and lowering operators. $\Delta = -\partial_x^2$ is the one-dimensional Laplacian and $g>0$ is a coupling constant which determines the strength of the interaction between the field and the matter. To obtain this Hamiltonian, we must assume that all field modes couple equally to each atom and ignore rapidly oscillating terms in the Rotating Wave Approximation (RWA) \cite{ScullyZubairy1997,CohenTannoudji1992}. See \cite{kraisler_22} for further discussion of this Hamiltonian in the case of a linear dispersion relation $\omega(k)=c|k|$ and \cite{KraislerLee2023} for the generalization to arbitrary smooth dispersion relation $\omega(k)$.\\

Quadratic dispersion relations arise naturally for excitations with an effective mass, including cavity photons in dye-filled microcavities \cite{Klaers2010}, semiconductor exciton-polaritons \cite{Weisbuch1992,Bieganska2026}, and slow-light  polaritons \cite{Gullans2016,Firstenberg2013,Baba2009}. More generally, polaritons arise through the hybridization of electromagnetic and material excitations \cite{Hopfield1958}; see \cite{CarusottoCiuti2013} for a broad review of quantum fluids of light and exciton-polariton systems.\\

As the number of total excitations is conserved by each term in the Hamiltonian $\mathcal{H}$, the linear space of single-excitation states is invariant. Such states have the form
\begin{align}
    \vert\Psi(t)\rangle = \int_{\R} \left[\psi(x,t)\phi^{\dagger}(x) + \alpha(x,t)\rho(x)\sigma^{\dagger}(x)\vert 0\rangle\right] dx.
\end{align}
where $\vert 0\rangle$ is the zero-excitation joint vacuum state of the field and atoms. Here $\psi(x,t)$ is the amplitude for the field excitation, while $\alpha(x,t)$ is the amplitude for an atomic excitation at position $x\in\R$. It can be shown, similarly to Appendix A in \cite{kraisler_22} that $\alpha$ and $\psi$ satisfy the pair of coupled evolution equations
\begin{subequations}\label{eq:time_dependent}
\begin{align}
i\partial_t\psi &=-\partial_x^2\psi+g\rho(x)\alpha,\\
i\partial_t\alpha & =g\psi+\Omega\alpha.
\end{align}
\end{subequations}
along with the normalization condition
\begin{align}
\int_{\R}\left(|\psi(x,t)|^2+|\alpha(x,t)|^2\rho(x)\right)dx=1.
\end{align}
Thus $\psi$ belongs to the usual space $L^2(\R)$ of square-integrable functions, while $\alpha$ belongs to the weighted space $L^2(\R,\rho(x)dx)$. This normalization allows us to interpret $|\psi(x,t)|^2$ as the probability density for finding the field excitation near $x$, and $|\alpha(x,t)|^2\rho(x)$ as the probability density for finding an excited atom near $x$.\\

For $\omega\in\R$, time harmonic solutions to the system \eqref{eq:time_dependent} of the form $\psi(x,t)=e^{-i\omega t}\psi(x)$, $\alpha(x,t)=e^{-i\omega t}\alpha(x)$ solve the following eigenvalue equation
\begin{subequations}\label{eq:spectral_problem}
\begin{align}
-\partial_x^2\psi+g\rho(x)\alpha &= \omega\psi,\label{eq:spectral_problem1}\\
g\psi+\Omega\alpha&=\omega\alpha\label{eq:spectral_problem2},
\end{align}
\end{subequations}
where $\psi\in L^2(\R)$ and $\alpha\in L^2(\R, \rho(x)dx)$. We call nontrivial solutions to the above problem \textit{bound states} with associated frequency $\omega$. The purpose of this paper is to study the existence of bound states for the system above for perturbations of constant background density. \\

Finally, we highlight the Hamiltonian viewpoint of the problem. If $\rho\in L^{\infty}$ and $\rho(x)\geq 0$, we can make the change of variables $\phi(x)=\sqrt{\rho(x)}\alpha(x)$ to see that \eqref{eq:spectral_problem} is equivalent to
\begin{align}
    H\begin{pmatrix}
        \psi \\ \phi 
    \end{pmatrix}=\omega\begin{pmatrix}
        \psi \\ \phi 
    \end{pmatrix},\qquad H = \begin{pmatrix}
        -\partial_x^2 & g\sqrt{\rho}\\ g\sqrt{\rho} & \Omega
    \end{pmatrix}.
\end{align}
The Hamiltonian, $H$, acts on the dense subspace of $L^2(\R;\C^2)$, the space of $\C^2$-valued square-integrable functions. The precise domain is $\mathcal{D}(H)=H^2(\R)\oplus L^2(\R)$, and bound states are exactly eigenvectors of this Hamiltonian.\\

The single-excitation problem considered here belongs to a broader class of continuum light–matter generalizing those considered by Dicke \cite{Dicke_1953}, Jaynes and Cummings \cite{Jaynes_1963}, and Tavis and Cummings \cite{TavisCummings_1968}, and related to generalized Friedrichs–Lee Hamiltonians \cite{Friedrichs1948,Lonigro2022}. Analytical and numerical aspects of single-excitation dynamics for discrete and continuously distributed atomic systems have been studied in \cite{kraisler_22,kraisler_23,Hoskins_21,Hoskins_23}. For the nonlocal photon-dispersion model, localized atomic densities were considered in \cite{NonlocPDE1}, while periodic atomic arrangements and their associated band structure were investigated in \cite{NonlocPDE2}. Bound states produced by spectral gaps also arise in the established literature on atoms coupled to photonic band-gap media \cite{JohnWang1990,JohnQuang1994}. The present work differs from these studies by considering a local effective-mass dispersion and exploiting the resulting ordinary differential equations to obtain general existence results and detailed, exactly solvable defect models.\\

The rest of the paper is organized as follows. In Section \ref{sec:constant_density} we show that the system \eqref{eq:spectral_problem} can be reduced to a scalar nonlinear eigenvalue problem and obtain the band structure of the constant density model. While this density does not support bound states, we find that the frequency spectrum has two bands and two gaps, and it is precisely within those gaps that bound states arise under perturbations of constant density. Following this, in Section \ref{sec:general_theorem} we state a general theorem on the existence of bound states for atomic densities which are sign-definite and spatially localized perturbations of a constant background. The proof, being more technical than the rest of the article, has been relegated to Appendix \ref{app:proof_main}. Sections \ref{sec:delta}, \ref{sec:double_delta}, and \ref{sec:square_well} make up the bulk of this work and study the problem of a single $\delta$-function, double $\delta$-function, and piecewise constant barriers and defects. In each case we determine a transcendental equation whose solutions determine the corresponding frequencies of bound states and analyze the number of such solutions as a function of the physical parameters in the model.

\section{Reduction to Scalar Problem and Constant Density}\label{sec:constant_density}
When $\omega\neq\Omega$ equation \eqref{eq:spectral_problem2} may be solved for $\alpha(x)$ as
\begin{align} 
\alpha(x)=\frac{g}{\omega-\Omega}\psi(x),
\end{align}
and the system \eqref{eq:spectral_problem} reduces to the scalar eigenvalue problem
\begin{align}\label{eq:reduced_spectral}
-\partial_x^2\psi+\frac{g^2\rho(x)}{\omega-\Omega}\psi=\omega\psi.
\end{align}
This is a nonlinear eigenvalue problem due to the presence of the spectral parameter $\omega$ in the effective potential. It may also be viewed as an eigenvalue problem for a one-dimensional Schr\"odinger operator with an energy dependent potential $V(x,\omega) = \dfrac{g^2\rho(x)}{\omega-\Omega}$.\\

It remains to consider whether this reduction loses possible solutions at the resonant frequency $\omega=\Omega$. For the densities studied in this paper, it does not. Indeed, if $\omega=\Omega$, then the equation \eqref{eq:spectral_problem2} implies $\psi(x)=0$ wherever $\rho(x)$ is supported. Since the densities considered here have support of full measure, this gives $\psi=0$ almost everywhere. Substituting back into equation \eqref{eq:spectral_problem2} then gives $\rho\alpha$, and hence $\alpha$, vanishes almost everywhere. Thus there are no nontrivial resonant bound states in the cases considered here.\\

Suppose that $\rho(x)\equiv\rho_0>0$ is a constant. Then the Hamiltonian $H$ has no bound states and the dispersion relation has two bands given by the solutions to the equation
\begin{align}
    \det\begin{pmatrix} k^2-\omega & g\rho_0\\ 
    g & \Omega - \omega
    \end{pmatrix}=0.
\end{align}
These bands never intersect and lie on either side of $\Omega$. Call them $\omega_{-}(k)< \Omega < \omega_{+}(k)$. Then
\begin{align}
    \omega_{\pm}(k) =\frac{1}{2}\left(\Omega + k^2 \pm \sqrt{(\Omega-k^2)^2+4g^2\rho_0}\right).
\end{align}
Hence, the spectrum is purely absolutely continuous and comprises two intervals
\begin{align}\label{eq:bands}
    I_{-}=\overline{\im(\omega_{-})}=[\omega_{-},\Omega],\qquad I_{+}=\overline{\im(\omega_{+})}=[\omega_{+},\infty).
\end{align} 
where $\omega_{\pm}=\omega_{\pm}(0)=\min_{k\in\R}\omega_{\pm}(k)$. Additionally, $\omega_{-}<0$ while $\omega_{+}>\Omega$. As we will see in the following sections, perturbing the constant density model leads to the formation of bound states appearing within the gaps
\begin{align}\label{eq:gaps}
    G_{-}= (-\infty ,\omega_{-}), \qquad G_{+}=(\Omega,\omega_{+}).
\end{align} 
For fixed $g>0$ and $\rho_0>0$, it is useful to define the following function $\lambda(\omega)$ 
\begin{align}
    \lambda(\omega) = \omega - \frac{g^2\rho_0}{\omega-\Omega},
\end{align}
whose sign determines whether $\omega$ lies in one of the intervals $I_{\pm}$ or gaps $G_{\pm}$. More specifically, we observe the following relationship between $\lambda(\omega)$ and the spectrum of $H$ in the constant density case.
\begin{figure}[ht]
    \centering
    \includegraphics[width=\linewidth]{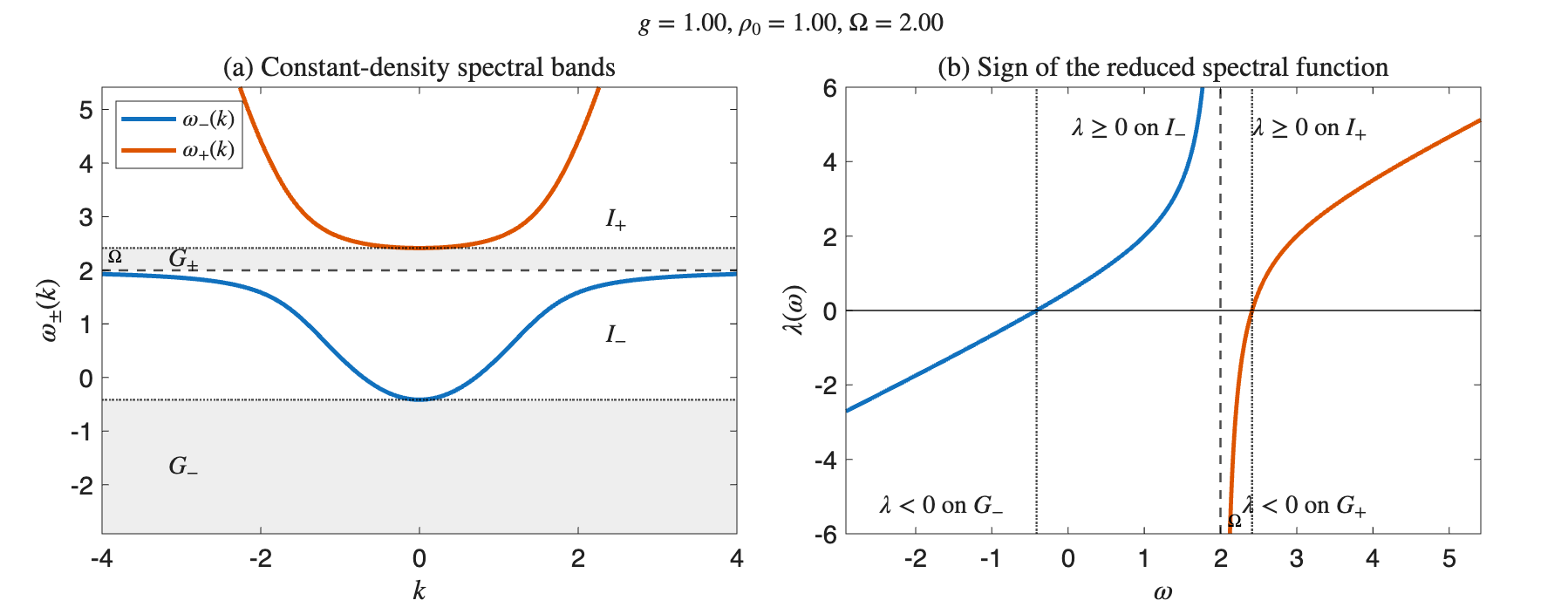}
    \label{fig:spectral_bands}
    \caption{(a) Spectral bands $\omega_{\pm}(k)$ for the constant density model (b) Plot of the spectral function $\lambda(\omega)$ and the intervals of constant sign.}
\end{figure}

\begin{itemize}
    \item $\lambda(\omega) = 0$ if and only if $\omega = \omega_{\pm}$.
    \item $\lambda(\omega)\in (0,\infty]$ if and only if $\omega\in I_{-}\cup I_{+}$.
    \item $\lambda(\omega) < 0$ if and only if $\omega\in G_{-}\cup G_{+}$.
\end{itemize}

\begin{remark}
We remark that by Weyl's Theorem \cite[Theorem XIII.14]{ReedSimonIV} compactly supported perturbations of $\rho_0$ do not change the essential spectrum of the operator $H$. Thus the gaps $G_{\pm}$ persist in all of the examples we consider. However, as we will see, bound states (i.e. eigenvalues) can appear in these gaps.
\end{remark}

\section{Bound states for generic sign-definite perturbations}\label{sec:general_theorem}
Before studying several exactly solvable models, we state a general result on the existence of solutions to the system \eqref{eq:spectral_problem}.
\begin{theorem}\label{thm:main_theorem}
Let $\rho(x)=\rho_0+\rho_1(x)$ where $\rho_0>0$ and $\rho_1\in L_c^{\infty}(\R)$ the space of compactly supported and essentially bounded functions. If $\rho_1\geq 0$ and $\rho_1\not\equiv 0$ then \eqref{eq:spectral_problem} has at least one solution in the lower gap $G_{-}$. If instead $\rho_1\leq 0$, $\rho_1\not\equiv 0$, and there exists $\rho_*>0$ such that $\rho(x)\geq \rho_*$, then \eqref{eq:spectral_problem} has at least one solution in the upper gap $G_{+}$. 
\end{theorem}

The proof of \Cref{thm:main_theorem} relies on the Birman-Schwinger principle \cite{Birman61,Schwinger61} and is given in \Cref{app:proof_main}. We remark that while this theorem does not directly apply to the $\delta$-function perturbations in Sections \ref{sec:delta} and \ref{sec:double_delta}, the conclusion is consistent with the results in the case that all $\delta$-functions are atomic barriers ($\mu > 0$).

\section{Delta perturbation}\label{sec:delta}
In this section we study the eigenvalue problem for constant densities perturbed by a single delta function. We interpret the positive coefficient case as an idealized localized collection of atoms, while the negative coefficient case can be seen as a signed defect in the density; for the standard theory of one-dimensional point interactions, see \cite{Albeverio2005}. Moreover, the sign of the delta function determines in which gap the bound state appears. Specifically, let the atomic density $\rho$ be given
\begin{align}
    \rho(x) = \rho_0 + \mu\delta(x).
\end{align}
Then the eigenvalue problem  \Cref{eq:reduced_spectral} can be recast as
\begin{align}
    -\partial_x^2\psi +\frac{g^2\mu}{\omega-\Omega}\delta(x)\psi = \lambda(\omega)\psi.
\end{align}
To have a solution we must have $\lambda(\omega) < 0$ and hence $\omega$ in one of the spectral gaps. Let $\kappa(\omega)=\sqrt{-\lambda(\omega)}$. Then the eigenvalue problem above has the solution
\begin{align}
    \psi(x) = Ae^{-\kappa(\omega)|x|}.
\end{align}
To determine the value of $\omega$ we integrate the eigenvalue problem on an interval centered at $x=0$ of width $\epsilon$ and let $\epsilon\to 0$. From this we obtain
\begin{align}
    -\psi'(0^+)+\psi'(0^{-}) + \frac{g^2\mu}{\omega-\Omega}\psi(0) =0,
\end{align}
where $f(0^\pm) :=\lim_{\epsilon\to 0}f(\pm \epsilon)$. This leads to the characterization that $\omega$ is an eigenvalue if and only if
\begin{align}\label{eq:delta_condition}
    \kappa(\omega)= \dfrac{g^2\mu}{2(\Omega-\omega)}.
\end{align}
We remark that $\kappa:G_{-}\cup G_{+}\to[0,\infty)$ and has the following properties:
\begin{enumerate}
    \item $\kappa(\omega)=0$ if and only if $\omega = \omega_{\pm}$.
    \item $\displaystyle\lim_{\omega\to -\infty}\kappa(\omega)=+\infty$ and $\displaystyle\lim_{\omega\to \Omega^{+}}\kappa(\omega)=+\infty$.
    \item $\kappa(\omega)$ is strictly decreasing on $G_{-}$ and $G_{+}$.
\end{enumerate}
\begin{figure}[ht]
    \centering
    \includegraphics[width=\linewidth]{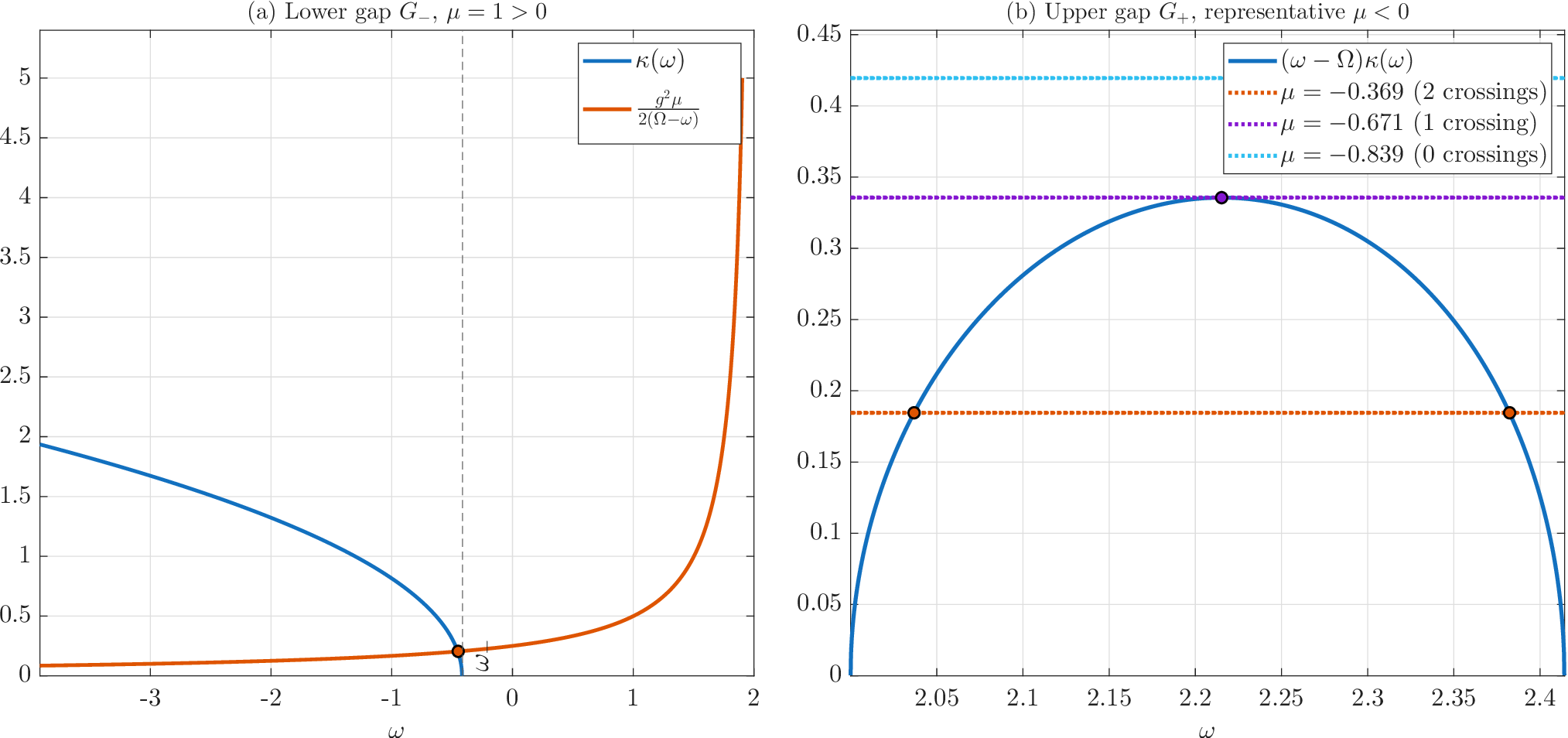}
    \label{fig:placeholder}
    \caption{Examples of crossings which correspond to solutions of \eqref{eq:delta_condition}. (a) A single intersection exists at a negative value. As $\mu$ increases, the orange curve increases and the intersection moves to the left. (b) $\dfrac{-g^2\mu}{2}$ is plotted against $(\omega-\Omega)\kappa(\omega)$ for three values of $\mu<0$. Note that $0$, $1$, or $2$ crossings may be achieved, and there are no solutions once $\mu<\mu_{c1}$. The values of the constants are $\Omega=2$, $\rho_0=1$, $g=1$.  }
\end{figure}
\subsection{Case 1: $\mu>0$}
Using the above properties of $\kappa(\omega)$, when $\mu>0$, the right hand side is negative for $\omega\in G_{+}$ while the left hand side is always nonnegative, hence there is no solution in the upper gap. However, on $G_{-}$ the function $\kappa(\omega)$ decreases from $\infty$ to $0$ while the right hand side increases from $0$ to $\infty$. Thus we have exactly one solution to \eqref{eq:delta_condition} in $G_{-}$. 
\subsection{Case 2: $\mu<0$}
When $\mu<0$ the situation is more subtle. As the right hand side of \eqref{eq:delta_condition} is negative on $G_{-}$, while the left hand side is always nonnegative, there can be no solution there. For $\omega\in G_{+}$ the condition is equivalent to
\begin{align}
    (\omega-\Omega)\kappa(\omega) = -\frac{g^2\mu}{2}.
\end{align}
The function on the left vanishes at each endpoint of $G_{+}$ and has negative second derivative in the interior. This implies it has a single non-degenerate maximum on $G_{+}$. We conclude that there can only be $0,1,$ or $2$ solutions depending on the value of $\mu$. There is a single critical value $\mu_{c}<0$ which gives $1$ solution, and for all $\mu<\mu_{c}$ there are no solutions. A direct calculation gives this critical value to be
\begin{align*}
    \mu_{c} = -\frac{2}{3\sqrt{3}g^2}(\sqrt{\Omega^2+3g^2\rho_0}-\Omega)\sqrt{\Omega + 2\sqrt{\Omega^2+3g^2\rho_0}}.
\end{align*}

We conclude that a local spike in atomic density ($\mu>0$) leads to an eigenvalue with strictly negative frequency, while a local absence in atomic density ($\mu<0$) can create up to two eigenvalues with frequency strictly larger than the resonant frequency $\Omega$. 

\section{Double Delta perturbation}\label{sec:double_delta}
Next we consider densities with two localized defects. Without loss of generality we may assume $\rho(x)$ is of the form
\begin{align}
    \rho(x)=\rho_0 + \mu_{+}\delta(x-d) + \mu_{-}\delta(x+d),
\end{align}
where the separation constant $d>0$. We will see that depending on the values and relationship between $\mu_{+}$ and $\mu_{-}$, there are between $0$ and $4$ eigenvalues with at most $2$ eigenvalues in each of the gaps $G_{-}$ and $G_{+}$.\\

Any bound state must once again have $\lambda(\omega)<0$ and hence have the form
\begin{align}
    \psi(x) = A_{+}e^{-\kappa(\omega)|x-d|}+A_{-}e^{-\kappa(\omega)|x+d|},
\end{align}
where, once again, $\kappa(\omega)=\sqrt{-\lambda(\omega)}$. Integrating around $x=\pm d$ as in the previous section gives two conditions
\begin{align}
    &-\psi'(+d^+)+\psi'(+d^{-}) + \frac{g^2\mu_{+}}{\omega-\Omega}\psi(+d) =0,\\
    &-\psi'(-d^+)+\psi'(-d^{-}) + \frac{g^2\mu_{-}}{\omega-\Omega}\psi(-d) =0.
\end{align}
This shows that $\omega$ is an eigenvalue if and only if there are nontrivial solutions to the following linear system
\begin{align}
    \begin{pmatrix}
        2\kappa + \beta_{+} & \beta_{+}e^{-2\kappa d}\\
        \beta_{-}e^{-2\kappa d} & 2\kappa+\beta_{-}
    \end{pmatrix}\begin{pmatrix}
        A_{+}\\ A_{-}
    \end{pmatrix}=\begin{pmatrix}
        0\\ 0
    \end{pmatrix},
\end{align}
where $\beta_{\pm}=\beta_{\pm}(\omega)=\dfrac{g^2\mu_{\pm}}{\omega-\Omega}$ and we have suppressed the $\omega$ dependence. Taking the determinant of this $2\times 2$ matrix gives a clean characterization of the eigenvalues $\omega$:
\begin{align}\label{eq:doubledelta_condition}
    (2\kappa + \beta_{+})(2\kappa + \beta_{-}) = \beta_{+}\beta_{-}e^{-4\kappa d}.
\end{align}
This is a transcendental equation, which may have up to four solutions. For simplicity, we consider only the case $|\mu_{+}|=|\mu_{-}|$. This substantially reduces the complexity of the calculations while already providing the primary phenomena.

\begin{figure}[ht]
    \centering
    \includegraphics[width=\linewidth]{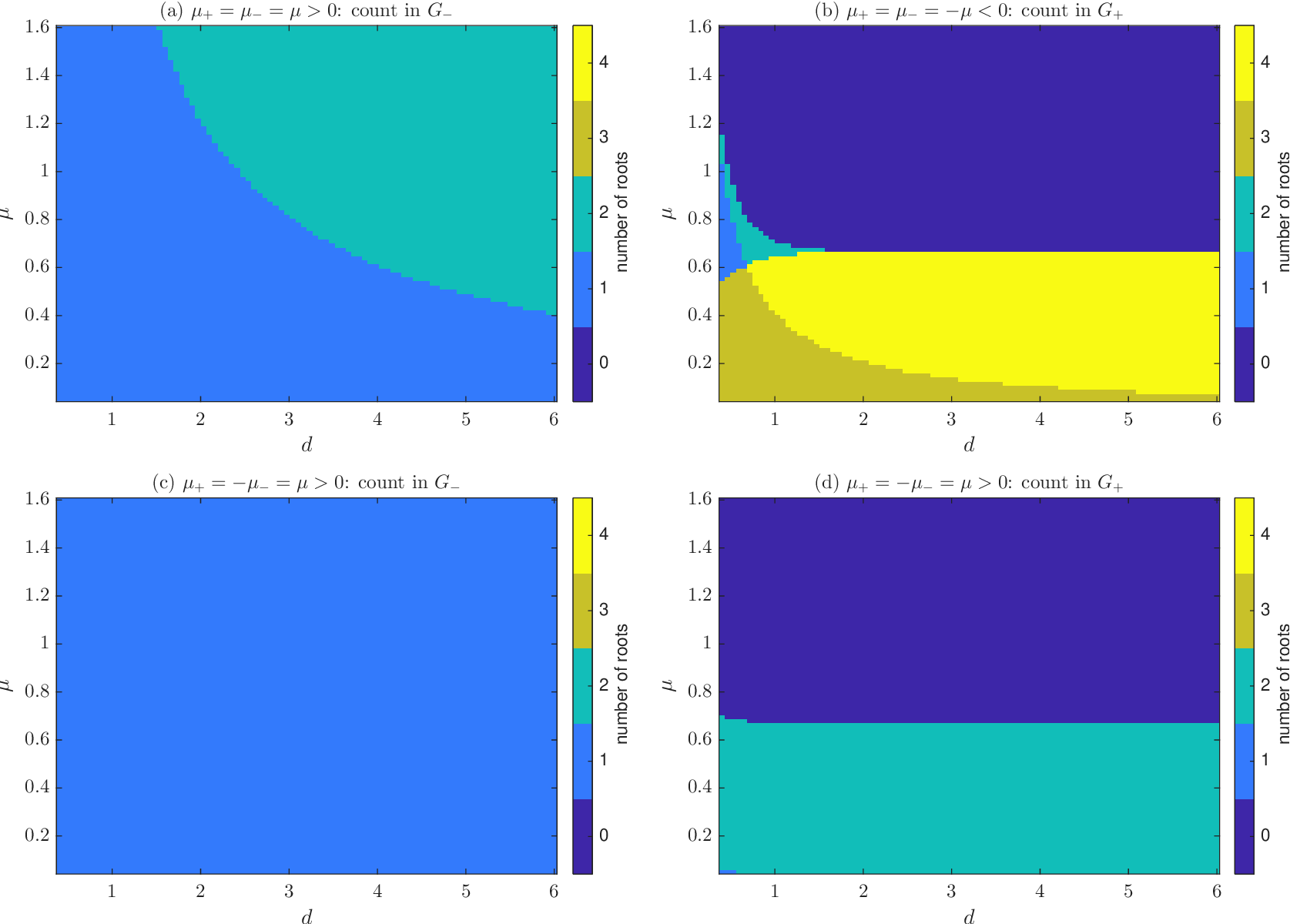}
    \caption{Heat maps showing the number of bound-states for various values of $d$ and $\mu$ in all cases. The number of roots is computed by subdividing each spectral gap, locating sign changes of the corresponding transcendental function, and separately detecting tangencies near double roots. All plots use the values $\Omega=2$, $\rho_0=1$, $g=1$. }
    \label{fig:double_delta}
\end{figure}
\subsection{Case 1: $\mu:=\mu_{+}=\mu_{-}$}
In this case the condition \eqref{eq:doubledelta_condition} becomes
\begin{align}
    (2\kappa + \dfrac{g^2\mu}{\omega-\Omega}(1+e^{-2\kappa d})) (2\kappa + \dfrac{g^2\mu}{\omega-\Omega}(1-e^{-2\kappa d}))=0.
\end{align}
Due to the fact that $\rho(x)$ is an even function, we may seek out even and odd solutions separately. It is not difficult to see that the conditions for even and odd solutions are:
\begin{itemize}
    \item Even: \ $\dfrac{\kappa(\omega)(\Omega-\omega)}{1+e^{-2\kappa(\omega)d}}=\dfrac{g^2\mu}{2}$
    \item Odd: \ $\dfrac{\kappa(\omega)(\Omega-\omega)}{1-e^{-2\kappa(\omega)d}}=\dfrac{g^2\mu}{2}$
\end{itemize}
For the analysis we introduce the auxiliary notation 
\begin{align*}
    Q_{e}(\kappa) = \frac{2\kappa}{1+e^{-2\kappa d}} = \kappa(1+\tanh(\kappa d)),\quad Q_{o}(\kappa) = \frac{2\kappa}{1-e^{-2\kappa d}} = \kappa(1+\coth(\kappa d)).
\end{align*}
Then the even and odd conditions become
\begin{align}\label{eq:doubledelta_condition_2}
    \mu = \frac{1}{g^2}Q_{e}(\kappa(\omega))(\Omega-\omega),\quad \mu = \frac{1}{g^2}Q_{o}(\kappa(\omega))(\Omega-\omega),
\end{align}
respectively. Note that $Q_{o},Q_{e}>0$. The identities
\begin{align*}
    Q_{e}'(\kappa) = 1 + \tanh^2(\kappa d)+ d\kappa \sech^2(\kappa d),\quad Q_{o}'(\kappa) = 1+\coth^2(\kappa d)-\kappa d\csch^2(\kappa d),
\end{align*}
show that these functions are increasing for all $\kappa > 0$.

\subsubsection{$\mu>0$} In this case all eigenvalues must be in the lower gap $G_{-}$. Differentiating the right hand sides of the conditions \eqref{eq:doubledelta_condition_2} we find
\begin{align*}
    Q_{e}'(\kappa(\omega))(\Omega-\omega)\kappa'(\omega) - Q_{e}(\kappa(\omega)) < 0, \\
    Q_{o}'(\kappa(\omega))(\Omega-\omega)\kappa'(\omega) - Q_{o}(\kappa(\omega)) < 0,
\end{align*}
as $(\Omega -\omega) > 0$ in $G_{-}$ and $\kappa'(\omega)<0$. Moreover, both curves tend to $+\infty$ as $\omega\to-\infty$. The curve $Q_{e}(\kappa(\omega))(\Omega-\omega)$ tends to $0$ as $\omega\to\Omega^{-}$ while 
\begin{align*}
    \lim_{\omega\to\Omega^{-}}\frac{1}{g^2}Q_{o}(\kappa(\omega))(\Omega-\omega) = \mu_{L}:=\dfrac{\Omega+\sqrt{\Omega^2+4g^2\rho_0}}{2dg^2},
\end{align*}
by direct calculation. Hence, the even condition always has a single solution which lies in the lower gap $G_{-}$, while the odd solution may have zero or one solution. If $\mu >\mu_{L}$ then there is one odd solution in $G_{-}$, while if $\mu\leq \mu_{L}$ there is no odd solution.

\subsubsection{$\mu<0$} Conversely, when $\mu<0$, all eigenvalues must be in the upper gap $G_{+}$. We claim that the conditions on the right hand side of \eqref{eq:doubledelta_condition_2} define concave up functions of $\omega$. Taking the second derivative of the right hand side of the conditions \eqref{eq:doubledelta_condition_2} we find

\begin{align*}
    \frac{(\kappa+2\omega-\Omega)^2(Q_j'-\kappa Q_j'')+4\kappa^2(\omega-\Omega)Q_j'}{4\kappa^3(\omega-\Omega)},\quad j = e,o.
\end{align*}
As $\kappa(\omega)>0$ and $(\omega-\Omega)>0$ on $G_{+}$ it suffices to show that $(Q_j'-\kappa Q_j'')>0$. We prove this inequality for $Q_{e}$ and omit the analogous argument for $Q_{o}$. 
\begin{align*}
    Q_e' -\kappa Q_{e}'' = 1+ \tanh(\kappa d) -\kappa d\sech^2(\kappa d) + 2d^2\kappa^2 \sech^2(\kappa d)\tanh(\kappa d).
\end{align*}
The last term is clearly positive, while the positivity of the first follows from
\begin{align*}
    1+ \tanh(\kappa d) -\kappa d\sech^2(\kappa d) &= \frac{e^{\kappa d}\cosh(\kappa d)-\kappa d}{\cosh^2(\kappa d)}=\frac{\frac{1}{2}(e^{2\kappa d}+1)-\kappa d}{\cosh^2(\kappa d)}>0.
\end{align*}
Additionally, both right hand sides are negative and approach $0$ as $\omega\to\Omega^{+}$. The limits as $\omega \to \omega_{+}^-$ can once again be calculated directly.
\begin{align*}
    \lim_{\omega \to \omega_{+}^-}\frac{1}{g^2}Q_{e}(\kappa(\omega))(\Omega-\omega)=0,\quad \lim_{\omega \to \omega_{+}^-}\frac{1}{g^2}Q_{o}(\kappa(\omega))(\Omega-\omega) =\mu_{R}:=\frac{\Omega-\omega_{+}}{dg^2}.
\end{align*}
Thus each right hand side has a unique minimum, which we call $\mu_{e}^*$ and $\mu_{o}^*$. By the fact that
\begin{align*}
    1 + e^{-2\kappa d}> 1-e^{-2\kappa d},
\end{align*}
we must have that $\mu_o^*<\mu_e^*<0.$ This implies that there are $0, 1,$ or $2$ even solutions as well as $0, 1, $ or $2$ odd solutions. For sufficiently negative values of $\mu$ there are no bound states. In Figure \ref{fig:double_delta}(b) we calculate the number of bound states numerically for various values of $\mu<0$ and $d>0$.

\subsection{Case 2: $\mu:=\mu_{+}=-\mu_{-}$}
In this case, we can no longer look for even and odd solutions. Instead we observe that the condition \eqref{eq:doubledelta_condition} reduces to 
\begin{align*}
    4\kappa(\omega)^2 -\frac{g^4\mu^2}{(\omega-\Omega)^2}(1-e^{-4\kappa(\omega)d})=0.
\end{align*}
Without loss of generality, we assume $\mu>0$. Then this may be recast as
\begin{align}\label{eq:double_delta_opposite}
    \mu = \frac{2\kappa(\omega)|\omega-\Omega|}{g^2\sqrt{1-e^{-4\kappa(\omega)d}}}.
\end{align}
A similar analysis to that in Case 1 shows that the right hand side of this equation is
\begin{itemize}
    \item positive for $\omega \in G_{-}\cup G_{+}$.
    \item monotone in $G_{-}$ decreasing from $+\infty$ to $0$.
    \item concave down with a unique maximum in $G_{+}$ vanishing at both end points.
\end{itemize}
Thus \eqref{eq:double_delta_opposite} has exactly one solution in the lower-gap $G_{-}$ for every $\mu>0$ and $0, 1,$ or $2$ solutions in the upper-gap $G_{+}$. Figure \ref{fig:double_delta}(c) and \ref{fig:double_delta}(d) show the numerical verification of these results.

\section{Finite Square Barrier}\label{sec:square_well}
In this section, we consider a piecewise constant density $\rho(x)$ which takes the value $\rho_1>0$ on a compact interval centered at the origin and the value $\rho_0\neq \rho_1$ outside of this interval. Specifically, fix $d > 0$ and define
\begin{align}
    \rho(x) = \begin{cases}
        \rho_1 & |x|\leq d\\
        \rho_0 & |x|> d
    \end{cases}.
\end{align}
Define $\lambda_{0}(\omega)$ and $\lambda_{1}(\omega)$ by
\begin{align}
    \lambda_{j}(\omega) = \omega - \frac{g^2\rho_j}{\omega-\Omega},\quad j = 0,1.
\end{align}
Then we seek solutions to the nonlinear eigenvalue problem \ref{eq:reduced_spectral}, repeated below for convenience.
\begin{align}
-\partial_x^2\psi+\frac{g^2\rho(x)}{\omega-\Omega}\psi=\omega\psi.
\end{align}
In order to obtain a square-integrable, continuously differentiable solution $\psi$, we must have exponentially decaying solutions in the region $|x|>d$ and oscillatory solutions in the region $|x|<d$. This leads to the pair of conditions
\begin{align*}
    \lambda_0(\omega) < 0, \quad \lambda_1(\omega)>0.
\end{align*}
We refer to the set of $\omega$ satisfying both conditions above as $\Sigma$. The first condition is simply that the eigenvalues must lie in the gaps $G_{-}\cup G_{+}$. The second condition is a genuinely new condition. To this end define
\begin{align}
    \omega_{j,\pm} = \frac{\Omega \pm\sqrt{\Omega^2 +4g^2\rho_j}}{2}.
\end{align}
Then we have the relations
\begin{itemize}
    \item $\lambda_0(\omega)< 0 \iff \omega \in (-\infty,\omega_{0,-})\cup(\Omega,\omega_{0,+})$
    \item $\lambda_1(\omega)> 0 \iff \omega \in (\omega_{1,-},\Omega)\cup (\omega_{1,+},\infty)$
\end{itemize}
Consequently, in order to obtain a solution $\omega$ must lie in one of the intervals
\begin{align*}
    (\omega_{1,-},\omega_{0,-}),\quad (\omega_{1,+}.\omega_{0,+}),
\end{align*}
with the convention that $(a,b)=\emptyset$ whenever $a\geq b$.
For $\omega\in\Sigma$ (i.e. satisfying $\lambda_0(\omega)<0$ and $\lambda_1(\omega)>0$), let $\kappa(\omega)$ and $q(\omega)$ be defined
\begin{align}
    \kappa(\omega)=\sqrt{-\lambda_0(\omega)},\quad q(\omega)=\sqrt{\lambda_1(\omega)}.
\end{align}
The fact that $\rho(x)$ is even allows us to once again seek out even and odd solutions separately. Any even solution to \eqref{eq:reduced_spectral} with $\rho(x)$ given above is of the form
\begin{align}
    \psi(x) = \begin{cases}
        A\cos(qx) & 0 < x < d \\
        Be^{-\kappa x} & x > d\\
    \end{cases}
\end{align}
Continuity of $\psi$ and $\psi'$ at $x=d$ lead to the pair of equations
\begin{align}
    &A\cos(qd) = Be^{-\kappa d},\\
    &Aq\sin(qd) = \kappa Be^{-\kappa d}.
\end{align}
Dividing one equation by the other we find that the condition for $\omega$ to be an eigenvalue associated to an even eigenfunction is
\begin{align}\label{eq:piecewise_even_condition}
    q(\omega)\tan(q(\omega)d)=\kappa(\omega).
\end{align}
A similar calculation leads to the following condition for $\omega$ to be an eigenvalue with corresponding to an odd eigenfunction.
\begin{align}\label{eq:piecewise_odd_condition}
    -q(\omega)\cot(q(\omega)d) = \kappa(\omega).
\end{align}
\subsection{Case 1: $\rho_1 > \rho_0$}
This corresponds to a localized barrier of additional atoms in the interval $[-d,d]$. In this case $\omega_{1,+}>\omega_{0,+}$ and hence we have that $\Sigma\cap G_{+}=\emptyset$. Thus any potential eigenvalues must lie in $(\omega_{1,-},\omega_{0,-})$. The function $q(\omega)$ is positive, continuous, increasing and satisfies
\begin{align}
    q(\omega_{1,-}) =0,\quad q(\omega_{0,-}) = K_{-}:=\sqrt{\frac{g^2(\rho_1-\rho_0)}{\Omega-\omega_{0,-}}},
\end{align}
while $\kappa(\omega)$ is decreasing. Hence for each integer multiple of $\pi/d$ that $q$ passes through, there is a solution to \eqref{eq:piecewise_even_condition}. Similarly, for each odd half-integer multiple of $\pi/d$ $q$ passes through, there is a solution to \eqref{eq:piecewise_odd_condition}. Consequently, the total number of eigenvalues is given  by the cardinality of the set
\begin{align}
    N_{-}=\left\{ m\in \Z_{\geq 0}\ | \ \frac{m\pi}{2}< K_{-}d\right\}.
\end{align}
We see that the number of eigenvalues is finite, yet monotonically increasing without bound in the coupling $g$, the barrier width $d$, and the density offset $\rho_1-\rho_0$. Plots of these curves and their intersections are given in Figure \ref{fig:finite_square}(a) and \ref{fig:finite_square}(c).

\subsection{Case 2: $\rho_1 < \rho_0$} This corresponds to a relative local absence of atoms in the interval $[-d,d]$. In this case $\omega_{0,-}<\omega_{1,-}$ and hence $\Sigma\cap G_{-}=\emptyset$. Any potential eigenvalues must lie in $(\omega_{1,+},\omega_{0,+})$. Additionally, $q(\omega)$ is still continuous, positive, and increasing, but now takes the values 
\begin{align}
     q(\omega_{1,+}) =0,\quad q(\omega_{0,+}) = K_{+}:=\sqrt{\frac{g^2(\rho_0-\rho_1)}{\omega_{0,+}-\Omega}}.
\end{align}
A similar argument shows that the number of eigenvalues is precisely
the cardinality of the set
\begin{align}
    N_{+}=\left\{ m\in \Z_{\geq 0}\ | \ \frac{m\pi}{2}< K_{+}d\right\},
\end{align}
which is once again monotonically increasing in $d$, $g$, and $\rho_0-\rho_1$. Again, plots of these curves and their intersections are given in Figure \ref{fig:finite_square}(b) and \ref{fig:finite_square}(d).

\begin{figure}[ht]
    \centering
    \includegraphics[width=\linewidth]{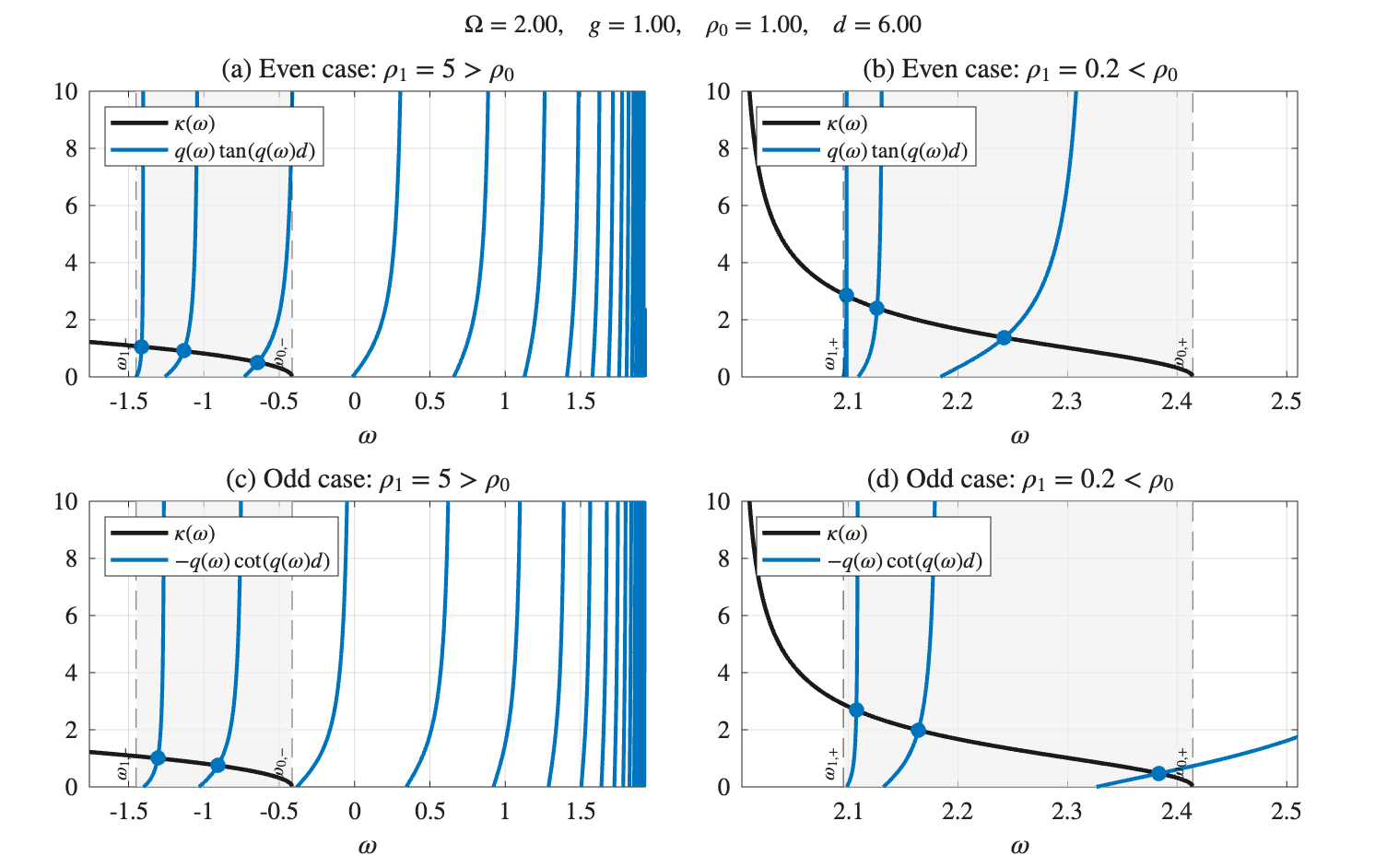}
    \caption{Plots of the curves in \eqref{eq:piecewise_even_condition} and \eqref{eq:piecewise_odd_condition} and their intersections. In (a) and (c) $\rho_1 =5 > \rho_0$ while in (b) and (d) $\rho_1=0.2<\rho_0$.}
    \label{fig:finite_square}
\end{figure}

\section{Discussion}
In this article we have studied a model of a one-dimensional scalar polaritonic field with a quadratic dispersion relation, $\omega(k)=k^2$, coupled to a continuum of identical two-level atoms with density $\rho(x)$ and resonant frequency $\Omega>0$. When $\rho(x)$ is constant, there are no bound states and the continuous (frequency) spectrum is composed of two bands. The lower band, $I_{-}$, is of the form $[\omega_{-},\Omega]$, where $\omega_{-}<0$, while the upper band, $I_{+}$, is unbounded above and of the form $[\omega_{+},\infty)$. Perturbations of the atomic density which are localized in space, either by local concentrations of atoms or negative density defects, lead to the formation of eigenvalues in the gaps between these two bands. \\

The task of determining bound state frequencies off resonance, $\omega\neq \Omega$, reduces to studying a scalar nonlinear eigenvalue problem for the polariton amplitude $\psi(x)$ which takes the form of one-dimensional Schr\"odinger operators with a frequency dependent effective potential $\dfrac{g^2\rho(x)}{\omega-\Omega}$. The physical mechanism can be summarized as follows: positive bounded defects lead to bound states in the lower gap $G_{-}$ while negative bounded defects create bound states in the upper gap $G_{+}$. In Theorem \ref{thm:main_theorem}, we proved this under mild hypotheses on the perturbation of constant density. Then we determined the number of bound states and associated energies exactly for several perturbations of constant background density including a single and double $\delta$-function concentration or defect as well as a piecewise constant square well or barrier. In the case of the $\delta$-function perturbations, there a limited number of bounds states arise regardless of the strength of the perturbation; however, in the case of the square well or barrier, the number of bound states grows with the respective difference between the defect and the background. \\

In order to derive this model we must ignore polarization, assume that all modes of the field couple to the atoms with equal strength $g>0$, and neglect rapidly oscillating terms through the rotating wave approximation. Restoring any one of these features would lead to a slightly more accurate model and is a topic for future study. Moreover, we only focus on bound states and questions of resonances and scattering states remain open. Additionally, the authors are also interested in analogous questions for photon-like fields with a linear dispersion relation $\omega(k)=c|k|$. This dispersion relation leads to a nonlocal system of equations characterized by a fractional Laplacian, and ordinary differential equation techniques are no longer available. 

\section*{Acknowledgements}
This work was supported in part by Amherst College through the Gregory S. Call Academic Intern program, which provided undergraduate research funding for SJ and RK.

\appendix

\section{Proof of \Cref{thm:main_theorem}}\label{app:proof_main}
In this section we present the proof of \Cref{thm:main_theorem}. We use a modified Birman-Schwinger principle \cite{Birman61,Schwinger61}, similar to that which appears in \cite{NonlocPDE1}. 
\begin{proof}
Let $\rho_0>0$ and $\rho_1\in L_c^{\infty}(\R)$. As discussed in the introduction, for $\omega\neq\Omega$ the system \eqref{eq:spectral_problem} reduces to the equation
\begin{align}
    \omega\psi=-\partial_x^2\psi+\frac{g^2\rho(x)}{\omega-\Omega}\psi.
\end{align}
We rewrite the equation as
\begin{align}
    (-\partial_x^2 -\lambda(\omega))\psi = -\frac{g^2\rho_1(x)}{\omega-\Omega}\psi.
\end{align}
For $\omega\in G_{-}\cup G_{+}$ we have $\lambda(\omega)<0$ and hence $(-\partial_x^2-\lambda(\omega))$ is invertible. Defining $\Psi = |\rho_1|^{1/2}\psi$ this equation is equivalent to
\begin{align}\label{eq:BS_principle}
    K(\omega)\Psi = \Psi,
\end{align}
where the operator $K(\omega)$ is defined
\begin{align}
    K(\omega) = \frac{\sgn(\rho_1) g^2}{\Omega-\omega}|\rho_1|^{1/2}(-\partial_x^2 -\lambda(\omega))^{-1}|\rho_1|^{1/2}.
\end{align}
Hence any solutions to \eqref{eq:spectral_problem} provides a solution to \eqref{eq:BS_principle}. Conversely, if $\Psi\in L^2$ is an eigenvector of $K(\omega)$ with eigenvalue $1$ we can define
\begin{align}
    \psi(x) =  \left(\frac{\sgn(\rho_1) g^2}{\Omega-\omega}(-\partial_x^2 -\lambda(\omega))^{-1}|\rho_1|^{1/2}\Psi\right)(x).
\end{align}
Note that $|\rho_1(x)|^{1/2}\psi(x)=K(\omega)\Psi = \Psi$ and that
\begin{align}
    (-\partial_x^2-\lambda(\omega))\psi = -\frac{g^2\rho_1}{\omega-\Omega}\psi,
\end{align}
which is the reduced equation \eqref{eq:reduced_spectral}. Moreover, $\psi\neq 0$ as $\Psi\neq 0$ and $\psi\in L^2(\R)$. Finally $\alpha(x)$ can be recovered by the relation $\alpha(x)=\frac{g}{\omega-\Omega}\psi(x)$. Hence we have the Birman-Schwinger type statement that: $\omega\neq \Omega$ is a bound-state frequency if and only if $1$ is an eigenvalue of $K(\omega$). \\

The operator $K(\omega)$ is self-adjoint on $L^2(\R)$ and has kernel $K(\omega; x,y)$ given by
\begin{align}
    K(\omega;x,y) = \frac{g^2\sgn(\rho_1)}{2\kappa(\omega)(\Omega-\omega)}|\rho_1(x)|^{1/2}e^{-\kappa(\omega)|x-y|}|\rho_1(y)|^{1/2}, \quad \kappa(\omega)=\sqrt{-\lambda(\omega)},
\end{align}
which is square integrable by the compact support of $\rho_1(x)$. This implies that the operator $K(\omega)$ is Hilbert-Schmidt and hence compact. For a discussion of Hilbert-Schmidt operators and their spectral theory see \cite[Chapter VI]{ReedSimonI}. Now suppose that $\rho_1\geq 0$ so that $\sgn(\rho_1)=1$ and let $\omega\in G_{-}$. Then on this set $\Omega-\omega>0$ and $K(\omega;x,y)>0$ and hence for $f\in L^2(\R)$
\begin{align}
    \langle f,K(\omega)f\rangle = \frac{g^2}{(\Omega-\omega)}\int_{\R} \frac{\left| |\widehat{(\rho_1^{1/2}f)}(\xi)\right|^2}{\xi^2+\kappa(\omega)^2}d\xi\geq 0,
\end{align}
which proves $K(\omega)$ is a positive operator. It suffices to show that there is some $\omega\in G_{-}$ such that $\|K(\omega)\|=1$ as the operator norm of a positive self-adjoint compact operator is always an eigenvalue. First we note that $\|K(\omega)\|$ is a continuous function of $\omega$ on $G_{-}$. Indeed, the kernel of $K(\omega)$ is a product of a continuous scalar prefactor, along with the exponential $e^{-\kappa(\omega)|x-y|}$ and $\omega$ independent terms $|\rho_1(x)|^{1/2}$, $|\rho_1(y)|^{1/2}$. This is enough to prove $K(\omega)$ is continuous in the Hilbert-Schmidt norm which implies $\|K(\omega)\|$ is continuous. We have
\begin{align}
    \lim_{\omega\to -\infty}\|K(\omega)\| = 0,\quad \lim_{\omega \to \omega_{-}^-}\|K(\omega)\| = \infty,
\end{align}
Thus by the continuity of $\|K(\omega)\|$ on $G_{-}$ the intermediate value theorem implies the existence of $\omega_*\in G_{-}$ such that $\|K(\omega_*)\|=1$. Similarly, if $\rho_1(x)\leq 0$ then $K(\omega;x,y)>0$ on $G_{+}$ as $\sgn(\rho_1)=-1$ and $\Omega-\omega<0$. On $G_{+}$ we have the limiting values
\begin{align}
    \limsup_{\omega\to \Omega^+}\|K(\omega)\| <1,\quad \lim_{\omega \to \omega_{+}^-}\|K(\omega)\| = \infty.
\end{align}
The bound as $\omega\to \Omega^+$ follows from the inequality
\begin{align}
    \|K(\omega)\|\leq \frac{-g^2}{(\omega-\Omega)\lambda(\omega)}\|\rho_1\|_{L^{\infty}} \to \frac{\|\rho_1\|_{L^{\infty}}}{\rho_0}, \text{as } \omega\to\Omega^+.
\end{align}
As $\rho_0+\rho_1(x)\geq \rho_*$ we have that this limiting ratio is strictly less than $1$. Thus, once again by the continuity of $\|K(\omega)\|$ (now on the set $G_{+})$ the intermediate value theorem implies there is an $\omega_*\in G_{+}$ such that $\|K(\omega_*)\|=1$.
\end{proof}

\printbibliography

\end{document}